\documentclass[review]{elsarticle}

\usepackage{textcomp}
\usepackage{amsmath,mathrsfs,amssymb}
\usepackage{amsthm} % definte the proof environment: block letters + italic letters
\usepackage[hidelinks,colorlinks=false]{hyperref}
\newtheorem{theorem}{Theorem}

\usepackage{lineno}
\modulolinenumbers[5]

\journal{elsevier}

\pdfoutput=1

\begin{document}

\begin{frontmatter}

\title{Backward bifurcation in SIRS malaria model}

%% Group authors per affiliation:
\author[1]{Miliyon Tilahun\corref{cor1}}
\ead{miliyon@ymail.com}

\begin{abstract}
We present a deterministic mathematical model for malaria transmission with waning immunity. The model consists of five non-linear system of differential equations. We used next generation matrix to derive the basic reproduction number $R_0$. The disease free equilibrium was computed and its local stability has been shown by the virtue of the Jacobean matrix. Moreover, using Lyapunov function theory and LaSalle Invariance Principle we have proved that the disease free equilibrium is globally asymptotically stable. Conditions for existence of endemic equilibrium point have been established. A qualitative study based on bifurcation theory reveals that backward bifurcation occur in the model. The stable disease free equilibrium of the model coexists with the stable endemic equilibrium when $R_0<1$. Furthermore, we have shown that bringing the number of disease (malaria) induced death rate below some threshold is sufficient enough to eliminate backward bifurcation in the model.
\end{abstract}


\end{frontmatter}

\linenumbers

\section{Introduction}

Malaria is an infectious disease caused by the parasitic infections of red blood cells by a protozoan of the genus \textit{Plasmodium} that are transmitted to people through the bites of infected female \textit{Anopheles} mosquitoes \cite{who}.
Malaria has for many years been considered as a global issue, and many epidemiologists and other scientists invest their effort in learning the dynamics of malaria and to control its transmission. From interactions with those scientists, mathematicians have developed a significant and effective tool, namely mathematical models of malaria, giving an insight into the interaction between the host and vector population, the dynamics of malaria, how to control malaria transmission, and eventually how to eradicate it. Mathematical modelling of malaria has flourished since the days of Ronald Ross \cite{Ross11}, who received the \textit{Nobel Prize} in Physiology or Medicine in $1902$ for his work on the life cycle of the malaria parasite. Ross developed a simple SIS-model.

\medskip

From the Ross's model, several models have been developed by researchers who extended his model by considering different factors such as latent period of infection in mosquitoes and humans \cite{12AronMay,Macdonald57}, age-related differential susceptibility to malaria in human population \cite{AronMay,12AronMay,Dietz}, acquired immunity \cite{AronMay,AronJL,Filipe}, and genetic heterogeneity of host and parasite \cite{Gupta1,Gupta2,Hasi,Rodrig,Torre}.

\medskip

In 2010, Yang et al. \cite{Yang} proposed SIR for the human and SI for the vector compartment model. But in their model, they assumed that the number of births for human and mosquito are independent of the total human and mosquito population. This assumption was later modified by Abadi and Krogstad \cite{Gebre} by making the number of births for human and mosquito dependent of the total human and mosquito population. However, Abadi and Krogstad made an assumption that, once the humans enter the recovered class they never go to the susceptible class again. Yet malaria does not confer permanent immunity \cite{Maia}. The recovered humans have a chance to be susceptible again.

\section{Model Formulation}

In this section, we formulate a mathematical model of malaria transmission with waning immunity. Because humans might repeatedly infected due to not acquiring permanent immunity so the human population is assumed to be described by the SIRS(Susceptible-Infected-Recovered-Susceptible) model. Mosquitoes are assumed not to recover from the parasites due to their short lifespan  so the mosquito population is described by the SI model. Recovered human hosts have temporary immunity that can be lost and are again susceptible to reinfection. All newborns are susceptible to infection, and the development of malaria starts when the infectious female mosquito bites the human host. The vectors do not die from the infection or are otherwise harmed. The flowchart of the model is shown in Figure \ref{fig:model}.

\begin{figure}[htb!]
\centering
\includegraphics[width=0.6\textwidth]{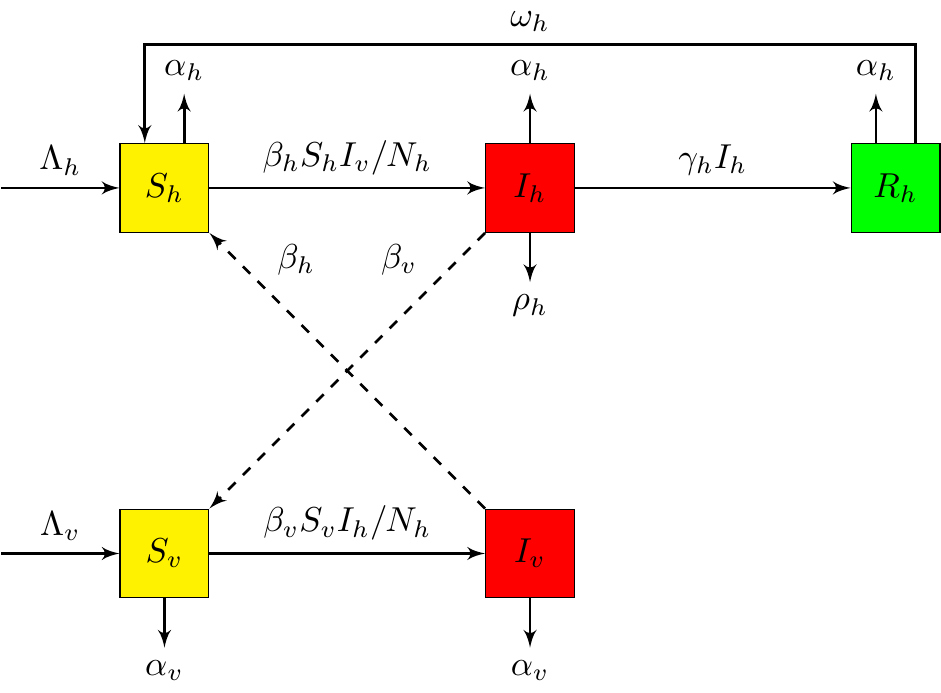}
\caption{Malaria model flowchart}\label{fig:model}
\end{figure}

The flowchart leads to the following system of non-linear ordinary differential equations
\begin{equation}\label{model1}
\begin{aligned}
  \frac{dS_h}{dt} &=\Lambda_h - \frac{\beta_h S_hI_v}{N_h}+\omega_h R_h-\alpha_h S_h, \\
  \frac{dI_h}{dt} &= \frac{\beta_h S_hI_v}{N_h}-\gamma_h I_h -\rho_h I_h-\alpha_h I_h,  \\
  \frac{dR_h}{dt} &=  \gamma_h I_h-\omega_h R_h-\alpha_h R_h, \\
  \frac{dS_v}{dt} &=  \Lambda_v-\frac{\beta_v S_vI_h}{N_h}-\alpha_v S_v, \\
  \frac{dI_v}{dt} &=  \frac{\beta_v S_vI_h}{N_h}-\alpha_v I_v,
\end{aligned}
\end{equation}

subjected to the initial conditions
\[
S_h(0)=S_{h0},\ I_h(0)=I_{h0},\ R_h(0)=R_{h0},\ S_v(0)=S_{v0},\ I_v(0)=I_{v0},
\]
where $N_h$ and $N_v$ represent the total size of the human population and mosquitoes population respectively. All the parameters can be found in Table \ref{table:2}.

\begin{table}[htb!]
\centering
\caption{State variables of malaria model}
\label{table:1}
\begin{tabular}{c l}
 \hline
 Symbol & Description   \\
 \hline
 $S_h$ & The number of susceptible human population  \\
% \hline
 $I_h$ & The number of infected human population  \\
% \hline
 $R_h$ & The number of recovered human population \\
% \hline
 $S_v$ & The number of susceptible mosquitoes population \\
% \hline
 $I_v$ & The number of infected mosquitoes population \\
 \hline
\end{tabular}
\end{table}
\begin{table}[htb!]
\centering
\caption{Parameters of malaria model}\label{table:2}
\begin{tabular}{ c l }
 \hline
 Symbol & Description   \\
 \hline
 $\Lambda_h$ & The recruitment rate of humans \\
% \hline
 $\alpha_h$ & The natural death rate of humans  \\
% \hline
 $\beta_h$ & The human contact rate \\
% \hline
 $\gamma_h$ & Per capita recovery rate for humans \\
% \hline
 $\rho_h$ & Per capita disease-induced death rate for humans\\
 $\omega_h$ & The per capita rate of loss of immunity in human \\
% \hline
 $\Lambda_v$ & The recruitment rate of mosquitoes  \\
% \hline
 $\alpha_v$ & The natural death rate of mosquitoes  \\
% \hline
 $\beta_v$ & The mosquito contact rate \\
 \hline
\end{tabular}
\end{table}

In the model, the term $ \frac{\beta_h S_hI_v}{N_h}$ denotes the rate at which the human hosts $S_h$ get infected by infected mosquitoes $I_v$ and $\frac{\beta_v S_vI_h}{N_h}$ refers to the rate at which the susceptible mosquitoes $S_v$ are infected by the infected human hosts $I_h$.

\medskip

The total population sizes $N_h$ and $N_v$ can be determined by $N_h=S_h +I_h +R_h$ and $N_v=S_v +I_v $ or from
the differential equations
\begin{align}
\frac{dN_h}{dt}&=\Lambda_h-\alpha_hN_h-\rho_h I_h,\label{totalh}\\
\frac{dN_v}{dt}&=\Lambda_v-\alpha_vN_v,\label{totalv}
\end{align}
which are derived by adding the first three equations of the system (\ref{model1}) for the human population and the
last two equations of the system (\ref{model1}) for mosquito vector population.

\subsection{Invariant Region and Positivity of Solutions}

The model represented by the system (\ref{model1}) will be analyzed in the feasible region and since we are working with
population all state variables and parameters are assumed to be positive. The invariant region can be obtained by the
following theorem.
\begin{theorem}[Invariant Region]
The solutions of the system (\ref{model1}) are feasible for all $t>0$ if they enter the invariant region
$\Omega =\Omega_h\times \Omega_v$, where
\begin{align*}
\Omega_h=\left\{(S_h,I_h,R_h)\in\Bbb R_+^3: S_h+I_h+R_h \leq\frac{\Lambda_h}{\alpha_h} \right\},
\end{align*}
\begin{align*}
\Omega_v=\left\{(S_v,I_v)\in\Bbb R_+^2: S_v+I_v \leq\frac{\Lambda_v}{\alpha_v} \right\}.
\end{align*}
\end{theorem}

Moreover, for the system (\ref{model1}) with a non-negative initial data to be epidemiologically meaningful and consistent, we need to show that all the state variables must remain non-negative $\forall t\geq 0$.

\begin{theorem}[Positivity]
Let the initial conditions of system (\ref{model1}) be positive. Then the solutions
$(S_h(t),I_h(t),R_h(t),S_v(t),I_v(t))$ of the system is non-negative for all  $t\geq 0$.
\end{theorem}

\section{Model Analysis}

\subsection{The Basic Reproduction Number}

The basic reproduction number($R_0$) is the average number of new infections, that one infected case will generate during their entire infectious lifetime \cite{Kenrad,Addo,heffernan2005}. It is very important in determining whether the disease persist in the population or die out.

\medskip

We use the next generation matrix to compute the basic reproduction number ($R_0$) which is formulated in \cite{Watmough}. Let us assume that there are $n$ compartments of which the first $m$ compartments correspond to infected individuals. Let
\begin{itemize}
\item $\mathscr{F}_i(x)$ be the rate of appearance of new infections in compartment $i$,
\item $\mathscr{V}_i^+(x)$ be the rate of transfer of individuals into compartment $i$ by all other means, and
\item $\mathscr{V}_i^-(x)$ be the rate of transfer of individuals out of compartment $i$.
\end{itemize}
It is assumed that each function is continuously differentiable at least twice in each variable. The disease transmission model consists of nonnegative initial conditions together with the following system of equations:
\begin{align}\label{R01}
\frac{dx_i}{dt}=f_i(x)=\mathscr{F}_i(x)-\mathscr{V}_i(x),\quad i=1,\ldots,n,
\end{align}
where $\mathscr{V}_i(x)=\mathscr{V}_i^-(x)-\mathscr{V}_i^+(x)$.

\medskip

The next step is the computation of the $m\times m$ square matrices $F$ and $V$ which are defined by
\[F=
\begin{bmatrix}
\frac{\partial\mathscr{F}_i }{\partial x_j}(x_0)
\end{bmatrix}
\mbox{ and }
V=
\begin{bmatrix}
\frac{\partial\mathscr{V}_i }{\partial x_j}(x_0)
\end{bmatrix}
\mbox{ with } 1\leq i,j\leq m.
\]
such that $F$ is non-negative, $V$ is a non-singular $M$-matrix and $x_0$ is the disease free equilibrium point (DFE) of (\ref{R01}). Since $F$ is non-negative and $V$ is non-singular, then $V^{-1}$ is nonnegative and also of $FV^{-1}$ is non-negative. The matrix $FV^{-1}$ is called the next generation matrix \cite{Watmough}. Finally, the basic reproduction number is given by
\[
R_0=\rho(FV^{-1}),
\]
where $\rho(A)$ denotes the spectral radius of a matrix $A$ and the spectral radius, $\rho(FV^{-1})$, is the largest absolute value of eigenvalues of the next generation matrix.

Thus for the model (\ref{model1}), $\mathscr{F}$ and $\mathscr{V}$ are given by
\begin{align}\label{R0f}
\mathscr{F}=
\begin{pmatrix}
\frac{\beta_h S_hI_v}{N_h}\\
\frac{\beta_v S_vI_h}{N_h}
\end{pmatrix},
\end{align}
and
\begin{align}\label{R0v}
\mathscr{V}=
\begin{pmatrix}
(\gamma_h +\rho_h +\alpha_h )I_h\\
\alpha_v I_v
\end{pmatrix}.
\end{align}
The disease free equilibrium (DFE) of the model is computed by setting the right hand side of the model (\ref{model1}) equal to zero with $I_h=0$ and $I_v=0$. Then the Jacobian matrix of (\ref{R0f}) at the disease free equilibrium (DFE) $E_0=\left(\frac{\Lambda_h}{\alpha_h},0,0,\frac{\Lambda_v}{\alpha_v},0\right)$ is given by
\[
F=
\begin{pmatrix}
0 & \beta_h\\
\frac{\beta_v \Lambda_v/\alpha_v}{\Lambda_h/\alpha_h} &0
\end{pmatrix}.
\]
Similarly, the Jacobian matrix of (\ref{R0v}) at the DFE $E_0$ is
\[
V=
\begin{pmatrix}
\gamma_h +\rho_h +\alpha_h & 0\\
0 & \alpha_v
\end{pmatrix}.
\]
Therefore, our next generation matrix is
\[
FV^{-1}=
\begin{pmatrix}
0 & \frac{\beta_h}{\alpha_v}\\
\frac{\beta_v\Lambda_v\alpha_h}{\alpha_v\Lambda_h(\gamma_h +\rho_h +\alpha_h)} & 0
\end{pmatrix}.
\]
Then the reproduction number $R_0=\rho(FV^{-1})$ is given by
\[
R_0=\sqrt{\frac{\beta_h\beta_v\Lambda_v\alpha_h}{\Lambda_h\alpha_v^2(\gamma_h +\rho_h +\alpha_h)}}.
\]
\subsection{Stability of Disease Free Equilibrium Point}

 The equilibria are obtained by equating the right hand side of the system (\ref{model1}) to zero. Disease-free equilibrium (DFE) of the model is the steady-state solution of the model in the absence of the disease (malaria). Hence, the DFE of the malaria model (\ref{model1}) is given by

\[
E_0=(S_h^0,I_h^0,R_h^0,S_v^0,I_v^0)=\left(\frac{\Lambda_h}{\alpha_h},0,0,\frac{\Lambda_v}{\alpha_v},0\right).
\]

\begin{theorem}
The DFE, $E_0$, of the system (\ref{model1}) is \textit{locally asymptotically stable} if $R_0<1$.
\end{theorem}

\begin{proof}
The Jacobian matrix of (\ref{model1}) at the disease-free equilibrium $E_0$ is
\begin{align}\label{jatdfe}
J(E_0)=
\begin{pmatrix}
-\alpha_h &  0                                                            & \omega_h             & 0         & -\beta_h  \\
0         & -(\gamma_h+\rho_h+\alpha_h)                                   & 0                    & 0         & \beta_h  \\
0         & \gamma_h                                                      & -(\alpha_h+\omega_h) & 0         &  0   \\
0         & -\frac{\beta_{v}\Lambda_{v}\alpha_{h}}{\alpha_{v}\Lambda_{h}} & 0                    & -\alpha_v & 0\\
0         & \frac{\beta_{v}\Lambda_{v}\alpha_{h}}{\alpha_{v}\Lambda_{h}}  & 0                    &  0        & -\alpha_v
\end{pmatrix}.
\end{align}
The characteristic equation of (\ref{jatdfe}) is given by
\[
(\lambda+\alpha_h)(\lambda+\alpha_v)(\lambda+\alpha_h+\omega_h)[\lambda^2+(k+\alpha_v)\lambda+\alpha_vk(1-R_0^2)]=0
\]
where $k=\gamma_h+\rho_h+\alpha_h$. The eigenvalues of (\ref{jatdfe}) are
\begin{align*}
\lambda_1&=-\alpha_h<0,\\ \lambda_2&=-\alpha_v<0,\\ \lambda_3&=-(\alpha_h+\omega)<0\\
\lambda_4&=\frac{-(k+\alpha_v)-\sqrt{(k+\alpha_v)^2-4\alpha_vk(1-R_0^2)}}{2}<0\\
\lambda_5&=\frac{-(k+\alpha_v)+\sqrt{(k+\alpha_v)^2-4\alpha_vk(1-R_0^2)}}{2}<0\mbox{ for }R_0<1
\end{align*}
Hence the result follows.
\end{proof}

\begin{theorem}
If $R_0^2<\Delta$, the disease free equilibrium ($E_0$) of the system (\ref{model1}) is \textit{globally asymptotically stable}, where
\begin{align*}
\Delta = \frac{\beta_h\beta_v\rho_h\Lambda_v}{\alpha_v^2\Lambda_h(\gamma_h+\rho_h+\alpha_h)}.
\end{align*}
\end{theorem}
\begin{proof}
To establish the global stability of the DFE ($E_0$), let
\begin{align*}
R_1 = 1-\frac{\beta_h\beta_v\rho_h\frac{\Lambda_v}{\alpha_v}}{\alpha_v\frac{\Lambda_h}{\alpha_h+\rho_h}(\gamma_h+\rho_h+\alpha_h)}.
\end{align*}
We notice that $R_1=R_0^2+1-\Delta$, thus $R_1<1$ is equivalent to $R_0^2<\Delta$. Since
\[
\frac{dN_h}{dt}=\Lambda_h-\alpha_h N_h-\rho_hI_h\geq \Lambda_h-(\alpha_h+\rho_h) N_h,
\]
we have
\[N_h\geq \frac{\Lambda_h}{\alpha_h+\rho_h}.\]
Now we choose the following Lyapunov function
\[
L(I_h,I_v)=\frac{\alpha_v}{\beta_h}I_h+I_v.
\]
Differentiating the Lyapunov function with respect to $t$, we obtain
\begin{align*}
\frac{dL}{dt} &= \frac{\alpha_v}{\beta_h}\frac{dI_h}{dt}+\frac{dI_v}{dt}\\
    &=\frac{\alpha_v}{\beta_h}\left(\frac{\beta_h S_hI_v}{N_h}-(\gamma_h+\rho_h+\alpha_h)I_h\right)+\frac{\beta_v S_vI_h}{N_h}-\alpha_v I_v\\
   &=\alpha_v\frac{S_hI_v}{N_h}-\frac{\alpha_v}{\beta_h}(\gamma_h+\rho_h+\alpha_h)I_h+\frac{\beta_v S_vI_h}{N_h}-\alpha_v I_v\\
   &=\left(\beta_v\frac{S_v}{N_h}-\frac{\alpha_v}{\beta_h}(\gamma_h+\rho_h+\alpha_h)\right)I_h-\alpha_v\left(1-\frac{S_h}{N_h}
   \right)I_v\\
   &\leq \left(\beta_v\frac{\Lambda_v/\alpha_v}{\Lambda_h/(\alpha_h+\rho_h)}-\frac{\alpha_v}{\beta_h}(\gamma_h+\rho_h+\alpha_h)\right)I_h\\
   &= \frac{\beta_h}{\alpha_v(\gamma_h+\rho_h+\alpha_h)}\left(R_1^2-1\right)I_h.
\end{align*}
Thus we have established that $\frac{dL}{dt}<0$ if $R_1<1$ and $\frac{dL}{dt}=0$ if and only if $I_h=0$, $I_v=0$. Therefore, the largest compact invariant set in
\[\left\{(S_h,I_h,R_h,S_v,I_v)\in \Omega:\frac{dL}{dt}=0\right\},\]
is the singleton set $\{E_0\}$ in $\Omega$. From LaSalle's invariant principle \cite{Khalil}, every solution that starts in the region $\Omega$ approaches $E_0$ as $t\to \infty$ and hence, the DFE $E_0$ is globally asymptotically stable for $R_0^2<\Delta$ in $\Omega$.
\end{proof}

\subsection{Existence and Stability of Endemic Equilibrium Point}

Endemic equilibrium point (EEP) is a steady state solution where the disease persists in the population. The EEP of (\ref{model1}) is given by $E_*=(S_h^*,I_h^*,R_h^*S_v^*,I_v^*)$ where
\begin{align*}
S_h^*=\frac{\alpha_v\Lambda_h+\alpha_h\beta_vI_h^*}{\alpha_h\alpha_vR_0^2},\ &
R_h^*=\frac{\gamma_h}{\alpha_h+\omega_h}I_h^*,\
S_v^*=\frac{\Lambda_v N_h}{\alpha_vN_h+\beta_vI_h^*},\\
I_v^*&=\frac{\beta_v\Lambda_vI_h^*}{\alpha_v(\alpha_vN_h+\beta_vI_h^*)},
\end{align*}
and $I_h^*$ is obtained by solving the equation
\begin{align}\label{polyendemic}
A(I_h^*)^2+BI_h^*+C=0,
\end{align}
where
\begin{equation}\label{polycoefficeint}
\begin{aligned}
A&=\alpha_h\alpha_v\beta_h\gamma_h\omega_hR_0^2N_h-\alpha_h\alpha_v\beta_h\beta_vR_{0v}
-\alpha_h^2\alpha_v\beta_v(\alpha_h+\omega_h)N_h^2,\\
B&=\alpha_h\alpha_v\beta_h\Lambda_h(\alpha_h+\omega_h)R_0^2N_h+\alpha_h\alpha_v^2\gamma_h\omega_hR_0^2N_h^2\\
&-\alpha_h\alpha_v\beta_h\Lambda_h(\alpha_h+\omega_h)N_h -\alpha_h^2\alpha_v\beta_v(\alpha_h+\omega_h)N_h^2
-\alpha_v^2\beta_h\Lambda_hR_{0v},\\
C&=\alpha_h\alpha_v^2\Lambda_h(\alpha_h+\omega_h)N_h^2\left(1-R_0^2\right).
\end{aligned}
\end{equation}

From (\ref{polycoefficeint}) it follows that $C>0$ whenever $R_0>1$. Thus, the number of possible positive real roots for (\ref{polyendemic}) depends on the signs of $A$ and $B$. This can be analyzed using the Descartes' Rule of Signs on the quadratic function
\begin{align*}
f(I_h^*)=A(I_h^*)^2+BI_h^*+C.
\end{align*}

The different possibilities for the roots $f(I_h^*)$ are tabulated in Table \ref{table:3}.

\begin{table}[htb!]
\centering
\caption{Number of possible positive real roots of $f(I_h^*)$ for $R_0<1$ and $R_0>1$.}
\label{table:3}
\begin{tabular}{c c c c c c c}
 \hline
 Cases & A & B & C & $R_0$ & \textnumero\ of sign changes & \textnumero\ of positive real roots   \\
 \hline
1      & + & + & + & $>1$            &          $0$           &  $0$ \\
2      & + & + & - & $<1$            &          $1$           &  $1$ \\
3      & + & - & + & $>1$            &          $2$           &  $2$ \\
4      & + & - & - & $<1$            &          $1$           &  $1$ \\
5      & - & + & + & $>1$            &          $1$           &  $1$ \\
6      & - & + & - & $<1$            &          $2$           &  $2$ \\
7      & - & - & + & $>1$            &          $1$           &  $1$ \\
8      & - & - & - & $<1$            &          $0$           &  $0$ \\
 \hline
\end{tabular}
\end{table}
Hence, we have established the following result
\begin{theorem}
The system (\ref{model1}) has a unique endemic equilibrium point where one of the cases $2,4,5$ and $7$ in Table \ref{table:3} are satisfied.
\end{theorem}

The existence of multiple endemic equilibrium point when $R_0<1$ is shown in Table \ref{table:3} which suggests the possibility of backward bifurcation \cite{Abba}, where the stable DFE coexists with a stable endemic equilibrium, when the
reproduction number is less than unity. Thus, the occurrence of a backward bifurcation has an important implications for epidemiological control measures, since an epidemic may persist at steady state even if $R_0<1$. This will be explored in the next section.

\subsection{Existence of Backward Bifurcation}

We shall use the following theorem in \cite{Song}, to show that the system (\ref{model1}) exhibits backward bifurcation at $R_0=1$.

\medskip

\begin{theorem}[Castillo-Chavez, Song]\label{thm:chavez}
Consider the following general system of ordinary differential equations with a parameter $\phi$
\begin{align}\label{bifur1}
\frac{dx}{dt}=f(x,\phi),\qquad f:\Bbb R^n\times \Bbb R\to\Bbb R  \mbox{ and } f\in \Bbb C^2(\Bbb R^n\times \Bbb R).
\end{align}
Without loss of generality, it is assumed that $0$ is an equilibrium for system (\ref{bifur1}) for all values of the parameter $\phi$, (that is $f(0,\phi)\equiv 0$). Assume
\begin{description}
  \item[(A1)] $A=D_xf(0,0)=\left(\frac{\partial f_i}{\partial x_j}(0,0)\right)$ is the linearized matrix of system (\ref{bifur1}) around the equilibrium $0$ with $\phi$ evaluated at $0$. Zero is a simple eigenvalue of $A$ and all other eigenvalues of $A$ have negative real parts;

  \item[(A2)] Matrix $A$ has a non-negative right eigenvector $w$ and a left eigenvector $v$ corresponding to the zero eigenvalue. Let $f_k$ be the $k^{\mbox{th}}$ component of $f$ and
  \begin{equation}
  \begin{aligned}
  a&= \sum_{k,i,j=1}^{n}v_kw_iw_j\frac{\partial^2f_k}{\partial x_i \partial x_j}(0,0),\\
  b&= \sum_{i,k=1}^{n} v_kw_i \frac{\partial^2f_k}{\partial x_i \partial \phi}(0,0).
  \end{aligned}
  \end{equation}
\end{description}
The local dynamics of system (\ref{bifur1}) around $0$ are totally determined by $a$ and $b$. Particularly, if $a > 0$ and $b > 0$, then a backward bifurcation occurs at $\phi= 0$.
\end{theorem}

To apply the above result, the following simplification and change of variables are made on the system (\ref{model1}). Let $S_h=x_1,I_h=x_2,R_h=x_3,S_v=x_4$ and $I_v=x_5$. So,
$N_h=x_1+x_2+x_3\mbox{ and } N_v=x_1+x_2$.
Moreover, by using vector notation $x=(x_1,x_2,x_3,x_4,x_5)^T$, the system (\ref{model1}) can be written in the form $\frac{dx}{dt}=(f_1,f_2,f_3,f_4,f_5)^T$ as follows
\begin{equation}\label{bifurmodel}
\begin{aligned}
  \frac{dx_1}{dt}=f_1 &=\Lambda_h - \frac{\beta_h x_1x_5}{x_1+x_2+x_3}+\omega_h x_3-\alpha_h x_1, \\
  \frac{dx_2}{dt}=f_2 &= \frac{\beta_h x_1x_5}{x_1+x_2+x_3}-(\gamma_h +\rho_h +\alpha_h) x_2, \\
  \frac{dx_3}{dt}=f_3 &=  \gamma_h x_2-\omega_h x_3-\alpha_h x_3, \\
  \frac{dx_4}{dt}=f_4 &=  \Lambda_v-\frac{\beta_v x_4x_2}{x_1+x_2+x_3}-\alpha_v x_4, \\
  \frac{dx_5}{dt}=f_5 &=  \frac{\beta_v x_4x_2}{x_1+x_2+x_3}-\alpha_v x_5, \\
\end{aligned}
\end{equation}
Choose $\beta_h=\beta_h^*$ as a bifurcation parameter. Solving for $\beta_h^*$ from $R_0=1$ gives
\begin{align*}
\beta_h^*=\frac{\Lambda_h\alpha_v^2(\gamma_h +\rho_h +\alpha_h)}{\beta_v\Lambda_v\alpha_h},
\end{align*}
The Jacobian matrix of the system (\ref{bifurmodel}) evaluated at the disease free equilibrium $E_0$ with $\beta_h=\beta_h^*$ is given by
\begin{align*}
J_*=
\begin{pmatrix}
  -\alpha_h  & 0        & \omega_h             & 0         & -\beta_h \\
  0          & -k       & 0                    & 0         & \beta_h \\
  0          & \gamma_h & -(\alpha_h+\omega_h) & 0         & 0 \\
  0          & -d       & 0                    & -\alpha_v & 0 \\
  0          &  d       & 0                    & 0         & -\alpha_v
\end{pmatrix},
\end{align*}
where $k=\gamma_h +\rho_h +\alpha_h$ and $d=(\beta_v\Lambda_v\alpha_h)/(\alpha_v\Lambda_h)$.

The Jacobian $J_*$ of the linearized system has a simple zero eigenvalue with all other eigenvalues having negative real part. For the case when $R_0=1$, using the technique in Castillo-Chavez and Song \cite{Song}, it can be shown that the matrix $J_*$ has a right eigenvector (corresponding to the zero eigenvalue), given by $w=[w_1\ w_2\ w_3\ w_4\ w_5]^T$, where
\begin{align*}
w_1=-\frac{\alpha_h^2+(\gamma_h+\omega_h +\rho_h)\alpha_h+\omega_h\rho_h}{\gamma_h\alpha_h}w_3,\quad w_2&=\frac{\omega_h+\alpha_h}{\gamma_h}w_3,\quad w_3=w_3>0,\\
w_4=-\frac{\alpha_h\beta_h\Lambda_v(\omega_h+\alpha_h)}{\gamma_h\alpha_v^2\Lambda_h}w_3,\quad w_5=&\frac{\alpha_h\beta_h\Lambda_v(\omega_h+\alpha_h)}{\gamma_h\alpha_v^2\Lambda_h}w_3.
\end{align*}
 Similarly, the components of the left eigenvector of $J_*$ (corresponding to the zero eigenvalue), denoted by $v=[v_1\ v_2\ v_3\ v_4\ v_5]$, are given by
\[
v_1=v_3=v_4=0,\quad v_2=v_2>0,\quad v_5=\frac{\Lambda_h\alpha_v(\gamma_h +\rho_h +\alpha_h)}{\beta_v\Lambda_v\alpha_h}v_2.
\]

\subsubsection*{Computation of $a$}
By computing the second-order partial derivatives at the disease free equilibrium point we have
\begin{align*}
\frac{\partial^2 f_2}{\partial x_1\partial x_j} & =0,\quad\mbox{for } j=1,2,3,4,5\\
\frac{\partial^2 f_2}{\partial x_2\partial x_j} & =0,\quad\mbox{for } j=1,2,3,4\\
\frac{\partial^2 f_2}{\partial x_3\partial x_j} & =0,\quad\mbox{for } j=1,2,3,4\\
\frac{\partial^2 f_2}{\partial x_4\partial x_j} & =0,\quad\mbox{for } j=1,2,3,4,5\\
\frac{\partial^2 f_2}{\partial x_5\partial x_j} & =0,\quad\mbox{for } j=1,4,5
\end{align*}
where as
\begin{align*}
\frac{\partial^2 f_2}{\partial x_2\partial x_5}=\frac{\partial^2 f_2}{\partial x_5\partial x_2}& =-\frac{\beta_h\alpha_h}{\Lambda_h},\\
\frac{\partial^2 f_2}{\partial x_3\partial x_5}=\frac{\partial^2 f_2}{\partial x_5\partial x_3}& =-\frac{\beta_h\alpha_h}{\Lambda_h}.
\end{align*}

 Similarly,
\begin{align*}
\frac{\partial^2 f_5}{\partial x_1\partial x_j} & =0,\quad\mbox{for } j=1,3,4,5\\
\frac{\partial^2 f_5}{\partial x_2\partial x_j} & =0,\quad\mbox{for } j=5\\
\frac{\partial^2 f_5}{\partial x_3\partial x_j} & =0,\quad\mbox{for } j=1,3,4,5\\
\frac{\partial^2 f_5}{\partial x_4\partial x_j} & =0,\quad\mbox{for } j=1,3,4,5\\
\frac{\partial^2 f_5}{\partial x_5\partial x_j} & =0,\quad\mbox{for } j=1,2,3,4,5
\end{align*}
while
\begin{align*}
\frac{\partial^2 f_5}{\partial x_1\partial x_2}&=\frac{\partial^2 f_5}{\partial x_2\partial x_1}
 =-\frac{\beta_v\Lambda_v\alpha_h^2}{\alpha_v\Lambda_h^2}=\frac{\partial^2 f_5}{\partial x_2\partial x_3}=\frac{\partial^2 f_5}{\partial x_3\partial x_2},\\
\frac{\partial^2 f_5}{\partial x_2^2}&=-\frac{2\beta_v\Lambda_v\alpha_h^2}{\alpha_v\Lambda_h^2},\
\frac{\partial^2 f_5}{\partial x_2\partial x_4}=\frac{\partial^2 f_5}{\partial x_4\partial x_2} =\frac{\beta_v\alpha_h}{\Lambda_h}.
\end{align*}
Then
\begin{align*}
a=&v_2\sum_{i,j=1}^{5}w_iw_j\frac{\partial^2f_2}{\partial x_i \partial x_j}(0,0)
+v_5\sum_{i,j=1}^{5}w_iw_j\frac{\partial^2f_5}{\partial x_i \partial x_j}(0,0)\\
=&2v_2\left(-w_2w_5\frac{\beta_h\alpha_h}{\Lambda_h}-w_3w_5\frac{\beta_h\alpha_h}{\Lambda_h}\right)\\
&+2v_5\left(-w_1w_2\frac{\beta_v\Lambda_v\alpha_h^2}{\alpha_v\Lambda_h^2}-w_2^2\frac{2\beta_v\Lambda_v\alpha_h^2}{\alpha_v\Lambda_h^2}
-w_2w_3\frac{\beta_v\Lambda_v\alpha_h^2}{\alpha_v\Lambda_h^2}+w_2w_4\frac{\beta_v\alpha_h}{\Lambda_h}\right)\\
=& \frac{2(\gamma_h+\rho_h+\alpha_h)(\alpha_h^2+(\gamma_h+\omega_h+\rho_h)\alpha_h+\omega_h\rho_h)(\omega_h+\alpha_h)v_2w_3^2}{
\Lambda_h\gamma_h^2}\\
&-\frac{4\alpha_h(\gamma_h+\rho_h+\alpha_h)(\omega_h+\alpha_h)^2v_2w_3^2}{\Lambda_h\gamma_h^2}-\frac{2\alpha_h(\omega_h+\alpha_h)
(\gamma_h+\rho_h+\alpha_h)v_2w_3^2}{\Lambda_h\gamma_h}\\
&-\frac{2\alpha_h\alpha_v\beta_h(\gamma_h+\rho_h+\alpha_h)(\omega_h+\alpha_h)^2v_2w_3^2}{\gamma_h^2}\\
=&\frac{2(\gamma_h+\rho_h+\alpha_h)(\omega_h+\alpha_h)}{\gamma_h}\biggl[\frac{\alpha_h^2+(\gamma_h+\omega_h+\rho_h)
\alpha_h+\omega_h\rho_h}{\Lambda_h\gamma_h} -\frac{2\alpha_h(\omega_h+\alpha_h)}{\Lambda_h\gamma_h}-\frac{\alpha_h}{\Lambda_h}\\
&-\frac{\alpha_h\alpha_v\beta_h(\omega_h+\alpha_h)}{\gamma_h}\biggl]v_2w_3^2\\
=& \frac{2(\alpha_v(\gamma_h+\rho_h+\alpha_h)(\omega_h+\alpha_h))^2}{\gamma_h^2\beta_h\beta_v\alpha_h\Lambda_v} \biggl[ \frac{\beta_v\Lambda_v\rho_h}{\Lambda_h^2\alpha_v^3(\gamma_h+\rho_h+\alpha_h)}-R_0^2\biggl]v_2w_3^2.
\end{align*}
\subsubsection*{Computation of $b$}
To compute $b$ we need to find the second order derivatives of $f_2$ and $f_5$ with respect to $x_i$ and $\beta_h$ at the disease free equilibrium point. Direct computation shows
\begin{align*}
\frac{\partial^2 f_2}{\partial x_i\partial \beta_h}&=0,\quad\mbox{for }i=1,2,3,4\\
\frac{\partial^2 f_5}{\partial x_i\partial \beta_h}&=0,\quad\mbox{for }i=1,2,3,4,5,
\end{align*}
and
\[
\frac{\partial^2 f_2}{\partial x_5\partial \beta_h}=\beta_h.
\]

\[
b=\sum_{i,k=1}^{5}v_kw_i\frac{\partial^2f_k}{\partial x_i \partial \beta_h}(0,0)=\beta_h v_2 w_5>0.
\]
Since $b$ is positive, it follows that the sign of $a$ determines the local dynamics around the disease free equilibrium for $\beta_h=\beta_h^*$. Based on Theorem \ref{thm:chavez}, system (\ref{model1}) will undergo backward bifurcation. Hence the following result holds.
\begin{theorem}
The malaria model (\ref{model1}) exhibits backward bifurcation at $R_0=1$ whenever $a$ is positive.
\end{theorem}

\begin{figure}[htb!]
\centering
\includegraphics[width=.6\textwidth]{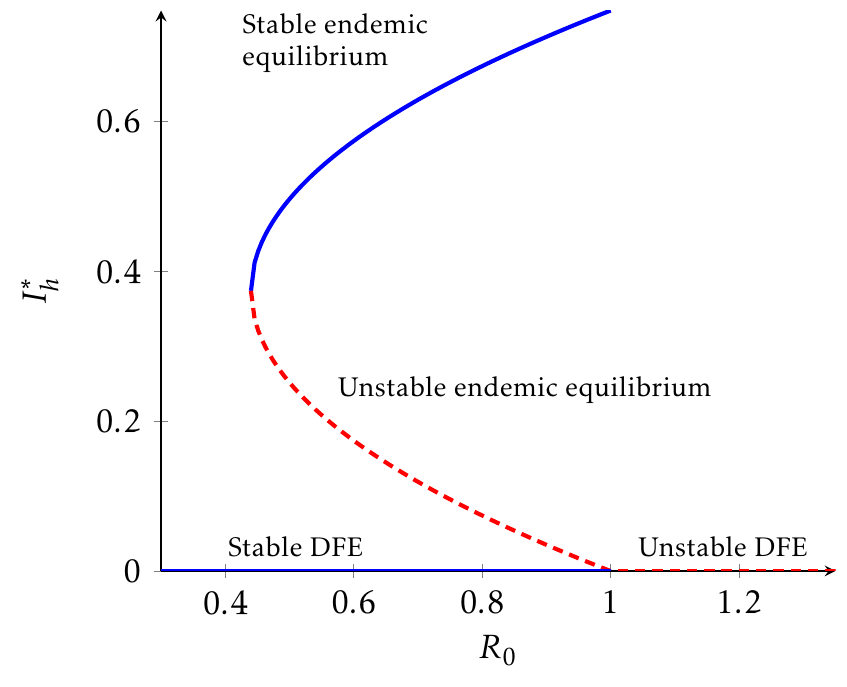}
\caption{Backward bifurcation phenomenon}\label{backbifur}
\end{figure}

The following parameter values were used to simulate the bifurcation diagram
$\Lambda_h=7, \Lambda_v= 0.071, \alpha_h=0.23, \alpha_v=0.54, \beta_v=0.82, \gamma_h= 0.143, \rho_h=0.27, \omega_h= 0.011$. It is important to note that these parameter values were used for illustrative purpose only, and may not be realistic epidemiologically.

Notice that $a>0\Leftrightarrow R_0<\sqrt{M}$ where $M=\frac{\beta_v\Lambda_v\rho_h}{\Lambda_h^2\alpha_v^3(\gamma_h+\rho_h+\alpha_h)}$. Moreover, $M=1$ gives
\[\rho_h=\frac{\Lambda_h^2\alpha_v^3(\gamma_h+\alpha_h)}{\beta_v\Lambda_v-\Lambda_h^2\alpha_v^3}:=\rho_h^*.\]
Hence, if the disease induced death rate satisfies $0\leq \rho_h\leq \rho_h^*$, then the disease can be eradicated provided that $R_0<1$.

\section{Numerical Simulations and Result}

In this section we present a numerical simulations of the model which is carried out using a fourth order Rung-Kutta scheme in
Matlab ode45. The values of the parameter used in the model are given in Table \ref{tab:par-value}.
\begin{table}[htb!]
\centering
\caption{Parameter values}\label{tab:par-value}
\begin{tabular}{p{5cm} p{5cm} r}
\hline
Parameter    &  Values  &  Reference           \\  \hline
$\Lambda_h$  & 2.5       & \cite{Lashari2011}  \\
$\beta_h$    & 0.01     &  \cite{Yang}         \\
$\alpha_h$   & 0.05     & \cite{Macdonald57}   \\
$\rho_h$     & 0.0001   &  Assumed             \\
$\gamma_h$   & 0.9      &  \cite{Gebre}        \\
$\Lambda_v$  & 500      & \cite{Lashari2011}   \\
$\beta_v$    & 0.005    & \cite{Yang}          \\
$\alpha_v$   & 0.06     &  \cite{Macdonald57}  \\
$\omega_h$   & 0.9      &  Assumed             \\ \hline
\end{tabular}
\end{table}

The initial conditions $S_h(0)=19413000 ,\ I_h(0)=3797000,\ R_h(0)=3790000,$ $ S_v(0)=16800000,\ I_v(0)=38200000$ were used for the simulations. In Figure \ref{humanplot}, the fractions of the populations, $S_h,I_h$ and $R_h$ are plotted versus time. The susceptible populations will initially decreases with time and then increases and the fractions of infected human populations decrease. The reproduction number is below one and the disease-free equilibrium point $E_0=(\Lambda_h/\alpha_h ,0,0)$ is stable. The susceptible and infected mosquito population decreases over time as shown in Figure \ref{moquitoplot}.

\begin{figure}[htb!]
\centering
\includegraphics[width=0.8\textwidth]{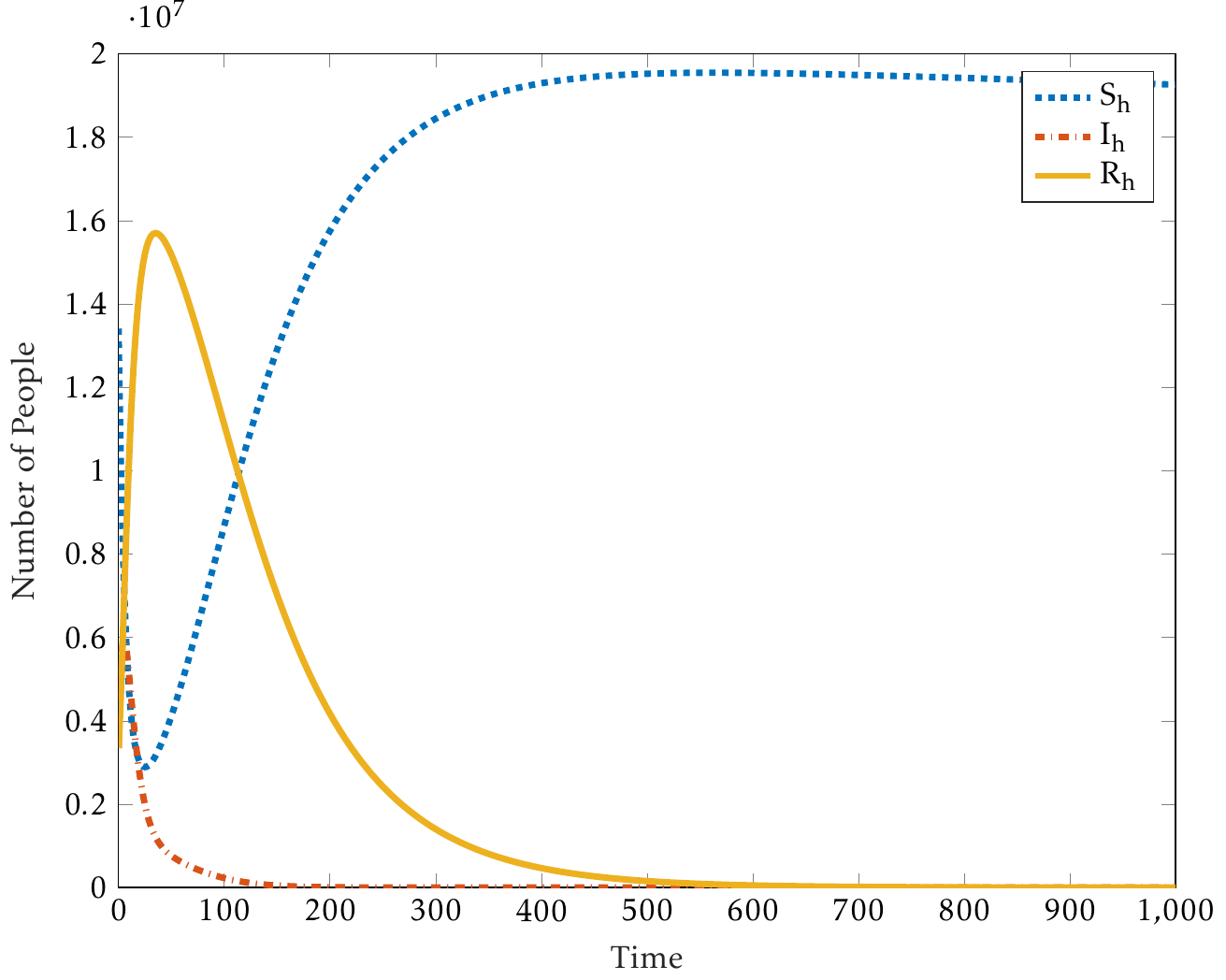}
\caption{Simulations of the model (\ref{model1}) for susceptible, infected and recovered human population which shows the local stability of the disease free equilibrium.}\label{humanplot}
\end{figure}

\begin{figure}[htb!]
\centering
\includegraphics[width=0.7\textwidth]{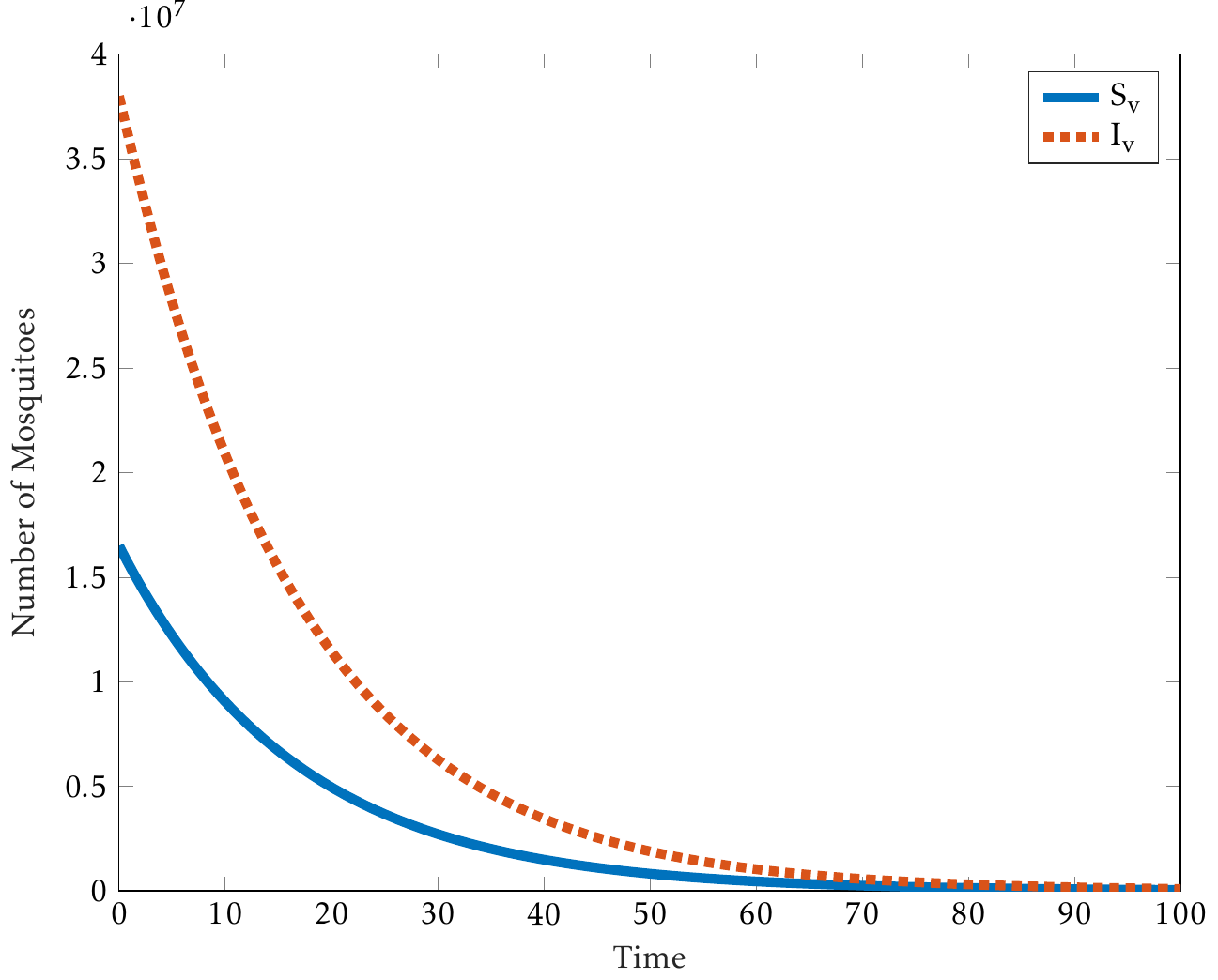}
\caption{Simulations of the model (\ref{model1}) for susceptible and infected mosquito population which shows the local stability of the disease free equilibrium.}\label{moquitoplot}
\end{figure}

\begin{figure}[htb!]
\centering
\includegraphics[width=0.7\textwidth]{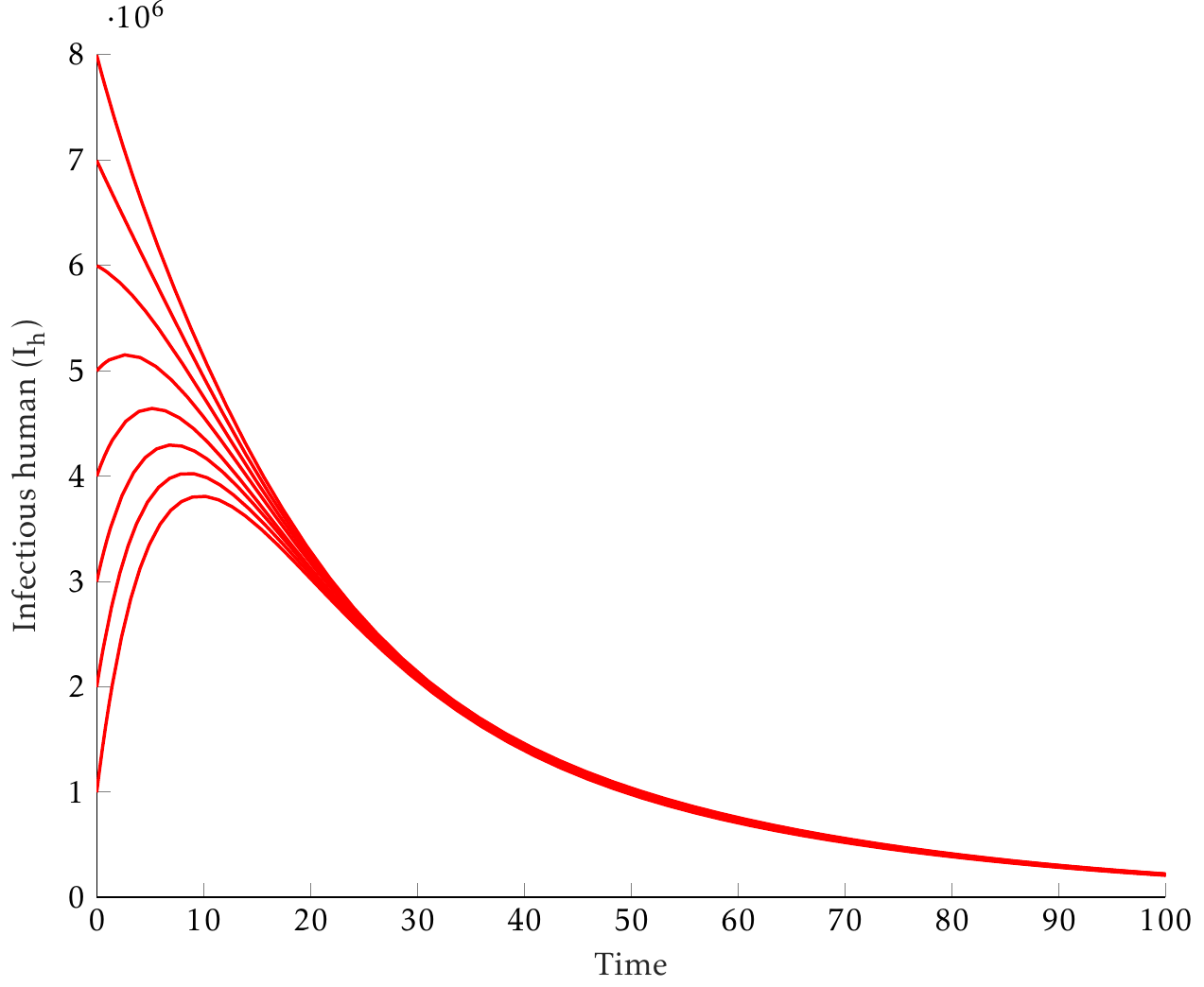}
\caption{Time series plot of the model (\ref{model1}) with different initial conditions which shows the global stability of the disease free equilibrium.} \label{humanseriesplot}
\end{figure}

\begin{figure}[htb!]
\centering
\includegraphics[width=0.8\textwidth]{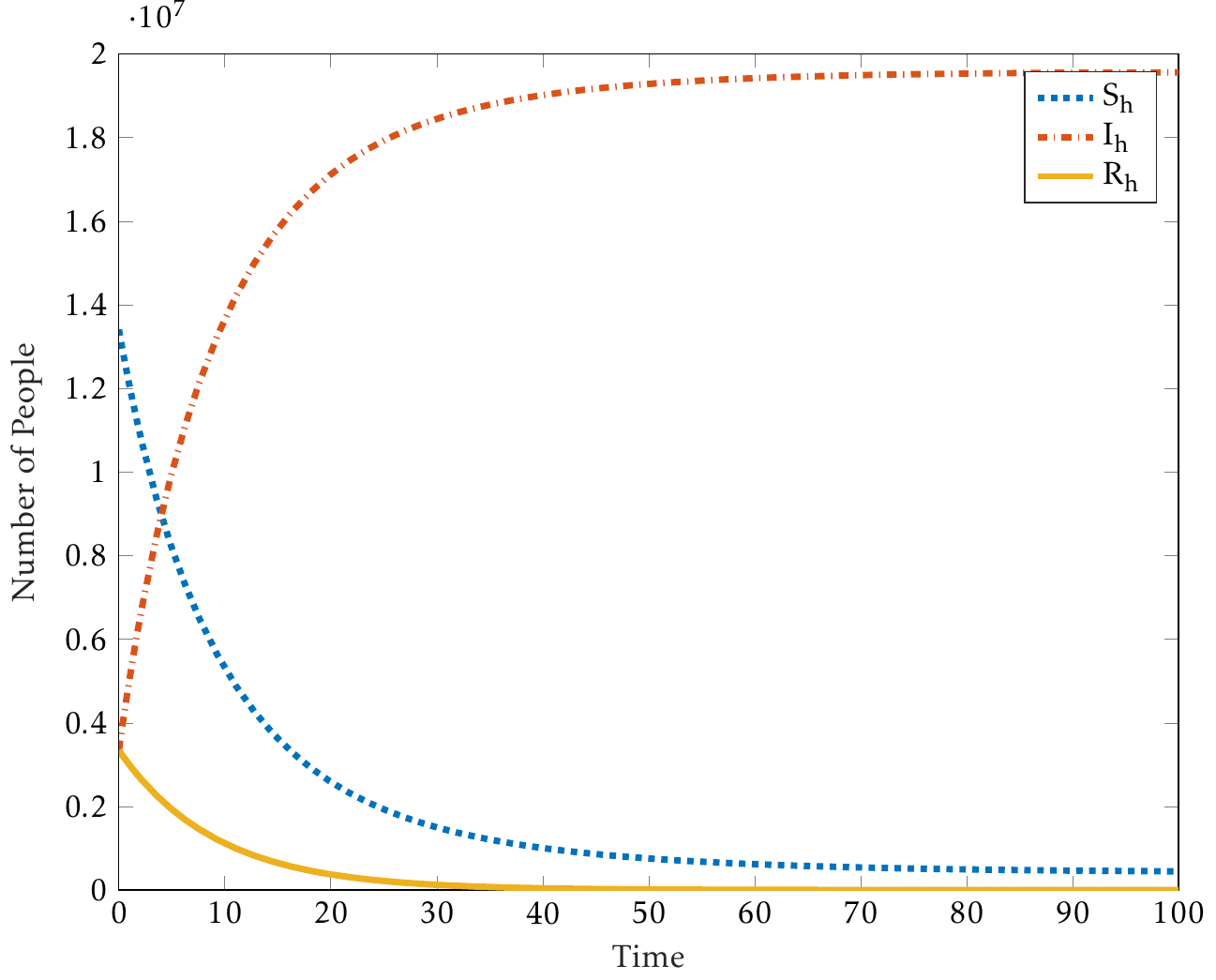}
\caption[Simulations of the model (\ref{model1}) for $R_0=47.6631$]{Simulations of the model (\ref{model1}) for $R_0=47.6631$ with parameter values $\Lambda_h=0.000091,\Lambda_v=0.071,\beta_h=0.0714285,\beta_v=0.09091,\gamma_h=0.000014285,\alpha_h=0.00004278,\alpha_v=0.04,
\rho_h=0.0000027,\omega_h=0.109$} \label{endemicsatbleplot}
\end{figure}

\begin{figure}[htb!]
\centering
\includegraphics[width=0.8\textwidth]{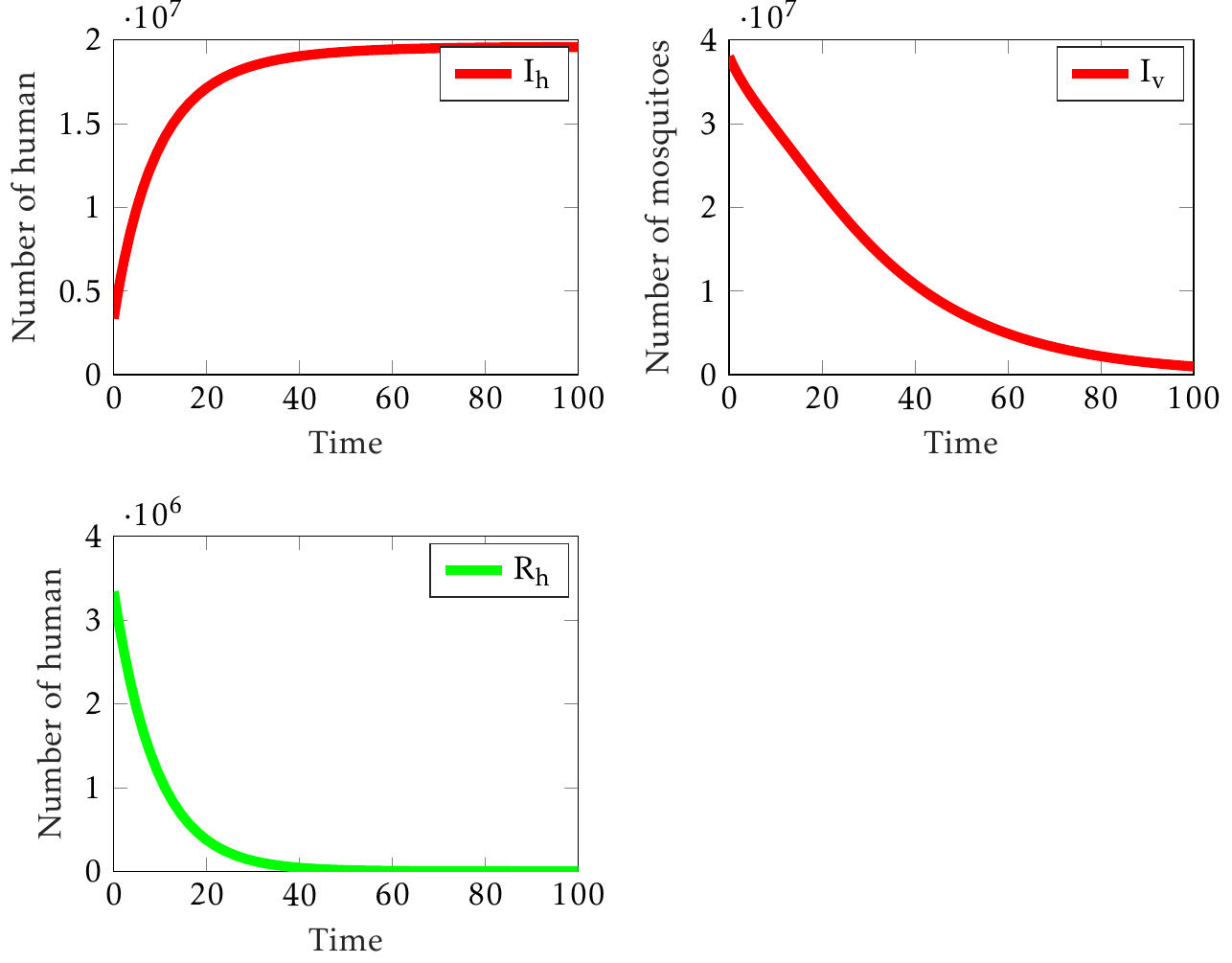}
\caption{Simulations of the model (\ref{model1}) for $R_0=47.6631$ showing infected human, infected mosquito and recovered human.  } \label{endemicsubplot}
\end{figure}

\section{Discussion}

In this paper, a deterministic mathematical model for malaria transmission has been presented. It was showed that there exists a domain in which the model is mathematically and epidemiologically well posed. The next generation matrix was used to derive the basic reproduction number $R_0$, which is the average number of new cases that one infected case will generate. The disease free equilibrium of the model was proved to be locally asymptotically stable whenever $R_0$ is less than unity. It is also showed that the disease free equilibrium is globally asymptotically stable provided that the basic reproduction number is less than some threshold. The unique endemic equilibrium point was shown to exist under certain conditions. The possibility of multiple endemic equilibrium point was discussed. It was shown that the model undergo backward bifurcation phenomenon. The stable disease free equilibrium coexist with the stable endemic equilibrium. Bringing the disease (malaria) induced death rate below some threshold was shown to be sufficient to eliminate backward bifurcation. Thus along with treated bed nets, and insecticides that would reduce the mosquito population there is a need for effective drug and efficient treatment which reduce the number of malaria induced death rate.

\newpage 

\bibliography{mybibfile}

\end{document}